\documentclass[11pt,a4paper]{article}
\usepackage{color,url}
\usepackage{algorithmicx}
\usepackage{algorithm}
\usepackage[noend]{algpseudocode}
\usepackage{amssymb,amsmath,amsthm,verbatim}
\newtheorem{thm}{Theorem}
\newtheorem{lem}[thm]{Lemma}
\newtheorem{defn}[thm]{Definition}
\newtheorem{cor}[thm]{Corollary}
\theoremstyle{definition}
\newtheorem{example}[thm]{Example}
\newcommand{\xc}{{\sc exact} \small 3\normalsize-{\sc cover}}
\newcommand{\sat}{{\sc \small (2,2)\normalsize-e\small 3\normalsize-sat}}
\def\MM{{\cal M}}
\def\PP{{\cal P}}
\begin{document}
\title{\bf Pareto optimal matchings of students to courses\\ in the presence of prerequisites\thanks{Supported by grant APVV-15-0091 from the Slovak Research and Development agency  (Cechl\'arov\'a), by the Swiss National Science Foundation SNFS (Klaus), by grant EP/K010042/1 from the Engineering and Physical Sciences Research Council (Manlove), and by a Short-Term Scientific Mission from COST Action IC1205 on Computational Social Choice (Manlove).  Part of this work was carried out whilst the authors visited Corvinus University of Budapest, and whilst the third author visited P.J. \v Saf\'arik University. This research was initiated during Dagstuhl Seminar 15241 on Computational Social Choice: Theory and Applications \cite{Dag15}.  We would like to thank the following Dagstuhl participants who were involved in initial discussions regarding the results in this paper: Haris Aziz, P\'eter Bir\'o, Jiehua Chen and Nicholas Mattei.
}}
\author{Katar\'\i na Cechl\'arov\'a$^1$, Bettina Klaus$^2$ and David F.\ Manlove$^3$\\ \\
\footnotesize\emph{$^1$Institute of Mathematics, Faculty of Science, P.J. \v Saf\'arik University,}\\
\footnotesize\emph{Jesenn\'a 5, 040 01 Ko\v sice, Slovakia.  Email {\tt katarina.cechlarova@upjs.sk}}.\\ \\
\footnotesize\emph{$^2$Faculty of Business and Economics, University of Lausanne, Internef 538,}\\
\footnotesize\emph{CH-1015 Lausanne, Switzerland.  Email {\tt Bettina.Klaus@unil.ch}.} \\ \\
\footnotesize\emph{$^3$School of Computing Science,  Sir Alwyn Williams Building, University of Glasgow,}\\
\footnotesize\emph{Glasgow, G12 8QQ, UK.  Email {\tt david.manlove@glasgow.ac.uk}.}}
\date{}
\maketitle
\begin{abstract}
We consider the problem of allocating applicants to courses, where each applicant has a subset of acceptable courses that she ranks in strict order of preference.  Each applicant and course has a \emph{capacity}, indicating the maximum number of courses and applicants they can be assigned to, respectively.  We thus essentially have a many-to-many bipartite matching problem with one-sided preferences, which has applications to the assignment of students to optional courses at a university.

We consider additive preferences and lexicographic preferences as two means of extending preferences over individual courses to preferences over bundles of courses.  We additionally focus on the case that courses have prerequisite constraints: we will mainly treat these constraints as compulsory, but we also allow alternative prerequisites.  We further study the case where courses may be corequisites.

For these extensions to the basic problem, we present the following algorithmic results, which are mainly concerned with the computation of Pareto optimal matchings (POMs).  Firstly, we consider compulsory prerequisites.  For additive preferences, we show that the problem of finding a POM is NP-hard.  On the other hand, in the case of lexicographic preferences we give a polynomial-time algorithm for finding a POM, based on the well-known sequential mechanism.  However we show that the problem of deciding whether a given matching is Pareto optimal is co-NP-complete.  We further prove that finding a maximum cardinality (Pareto optimal) matching is NP-hard.  Under alternative prerequisites, we show that finding a POM is NP-hard for either additive or lexicographic preferences.  Finally we consider corequisites.  We prove that, as in the case of compulsory prerequisites, finding a POM is NP-hard for additive preferences, though solvable in polynomial time for lexicographic preferences.  In the latter case, the problem of finding a maximum cardinality POM is NP-hard and very difficult to approximate.
\end{abstract}

\section{Introduction}\label{s_intro}
Problems involving the allocation of indivisible goods to agents have gained a lot of attention in the literature, since they model many real scenarios, including the allocation of pupils to study places \cite{BS99}, workers to positions \cite{KC82}, researchers to projects \cite{MT13}, tenants to houses \cite{AS98} and students to courses \cite{BC12}, etc.  We assume that agents on one side of the market (pupils, workers, researchers, tenants, students) have preferences over objects on the other side of the market (study places, positions, projects, courses, etc.)\ but not vice versa.  In such a setting where the desires of agents are in general conflicting, economists regard Pareto optimality (or Pareto efficiency) as a basic, fundamental criterion to be satisfied by an allocation.  This concept guarantees that no agent can be made better off without another agent becoming worse off.  A popular and very intuitive approach to finding Pareto optimal matchings is represented by the class of \emph{sequential allocation mechanisms} \cite{KC71,BK05,BL11,AWX15}.  

In the one-to-one case (each agent receives at most one object, and each object can be assigned to at most one agent) this mechanism has been given several different names in the literature, including serial dictatorship \cite{AS98,Man07}, queue allocation \cite{Sve94}, Greedy-POM \cite{ACMM04} and sequential mechanism \cite{BL11,AWX15}, etc.  Several authors independently proved that a  matching  is Pareto optimal if and only if can be obtained by the serial dictatorship mechanism (Svensson in 1994 \cite{Sve94}, Abdulkadiro\v glu and S\"onmez in 1998 \cite{AS98}, Abraham et al.\ in 2004 \cite{ACMM04}, and Brams and King in 2005 \cite{BK05}). 

In general many-to-many matching problems (agents can receive more than one object, and objects can be assigned to more than one agent), the sequential allocation mechanism works as follows: a central authority decides on an ordering of agents (often called a \emph{policy}) that can contain multiple copies of an agent (up to her capacity).
According to the chosen policy, an agent who has her turn chooses her most preferred object among those that still have a free slot.  This approach was used in \cite{AWX15,BL11}, where the properties of the obtained allocation with respect to the chosen policy and strategic issues are studied. 

The serial dictatorship mechanism is a special case of the sequential allocation mechanism where the policy contains each agent exactly once, and when agent $a$ is dealt with, she chooses her entire most-preferred bundle of objects.
The difficulty with serial dictatorship is that it can output a matching that is highly unfair.  For example, it is easy to see that if there are two agents, each of whom finds acceptable all objects and has capacity equal to the number of objects, and each object has capacity 1, then the serial dictatorship mechanism will assign all objects to the first agent specified by the policy and no object to the other agent.

In this paper we shall concentrate on one real-life application of this allocation problem that arises in education, and so our terminology will involve \emph{applicants} (students) for the agents and \emph{courses} for the objects.  In most universities students have some freedom in their choice of courses, but at the same time they are bound by the rules of the particular university.  A detailed description of the rules of the allocation process and the analysis of the behaviour of students at Harvard Business School, based on real data, is provided by Budish and Cantillon \cite{BC12}. They assume that students have a linear ordering of individual courses and their preferences over bundles of courses are responsive to these orderings.  The emphasis in \cite{BC12} was on strategic questions. The empirical results confirmed the theoretical findings that, loosely speaking, dictatorships (where students choose one at a time their entire most preferred available bundle) are the only mechanisms that are strategy-proof and ex-post Pareto efficient.

Another field experiment in course allocation is described by Diebold et al.\ \cite{DABMS14}. The authors compared the properties of allocations obtained by the sequential allocation mechanism where the policy is determined by the arrival time of students (i.e., first-come first-served) and by two modifications of the Gale-Shapley student-optimal mechanism, i.e., they assumed  that courses may also have preferences or priorities  over students.  Moreover, they only considered the case when each student can be assigned to at most one course.

In reality, a student can  attend more courses, but not all possible bundles are feasible for her.  
Cechl\'arov\'a et al.\ \cite{CEFMMP14} considered explicitly-defined notions of feasibility for bundles of courses.  For these feasibility concepts, a given bundle can be checked for feasibility for a given applicant in time polynomial in the number of courses.  Such an algorithm may check for example if no two courses in the bundle are scheduled at the same time, or if the student has enough budget to pay the fees for all the courses in the bundle, etc.  Cechl\'arov\'a et al.\ \cite{CEFMMP14}  found out that a sufficient condition for a general sequential allocation mechanism to output a Pareto optimal matching is that feasible bundles of courses form families that are closed with respect to subsets, and preferences of students over bundles are lexicographic.  They also showed that under these assumptions a converse result holds, i.e., each Pareto optimal matching can be obtained by sequential allocation mechanism if a suitable policy is chosen. 

\subsubsection*{Prerequisites and corequisites}
In this paper we deal with prerequisite and corequisite constraints. Prerequisite constraints model the situation where a student may be allowed to subscribe to a course $c$ only if she subscribes to a set $C'$ of other course(s). The courses in $C'$ are usually called \emph{prerequisite courses}, or \emph{prerequisites}, for $c$. For example, at a School of Mathematics, an Optimal Control Theory course may have as its prerequisites a course on Differential Equations as well as a course in Linear Algebra; a prerequisite for a Differential Equations course could be a Calculus course, etc.  On the other hand, corequisite constraints model the situation where a student takes course $c_1$ if and only if she takes course $c_2$.  These courses are referred to as \emph{corequisite courses}, or \emph{corequisites}.  For example, a corequisite constraint may act on two courses: one that is a theoretical programming course and the other that is a series of corresponding programming lab sessions.

We consider three different models involving prerequisite and corequisite constraints. The first model involves compulsory prerequisites, 
but we allow for the possibility that different students may have different prerequisite constraints. For example, for a doctoral study program in Economics, an economics graduate may have as a prerequisite a mathematical course and, on the other hand, a mathematics graduate may have as a prerequisite a course on microeconomics, etc.
The second model concerns \emph{alternative prerequisites}.  Here it is assumed that certain courses require that a student subscribes to at least one of a set of other courses.  For example, a course in mathematical modelling may require that a student attends one of a range of courses that deal with a specific mathematical modelling software package, such as Maple, MATLAB or Mathematica. 
Finally, the third model considers corequisites.  Here we assume that constraints on corequisite courses are identical for all applicants.

As we assume that applicants express their preferences only over individual courses, a suitable extension of these preferences to preferences over bundles of courses has to be chosen. Among the most popular extensions are \emph{responsive} preferences \cite{Rot85a}.  That is, an applicant has responsive preferences over bundles of courses if bundle $C'$ is preferred to bundle $C''$ whenever $C'$ is obtained from $C''$ by replacing some course in $C''$ by a more preferred course not contained in $C''$.
Responsiveness is a very mild requirement and responsive preferences form a very wide and variable class. Therefore we shall restrict our attention to two specific examples, namely \emph{additive} \cite{BBP04,BL11,BC12} and \emph{lexicographic} \cite{Fis75,BBP04,SV12,SS13a,TSY14} preferences.
Although lexicographic preferences can be modelled as additive preferences by choosing appropriate weights \cite{BL11}, we would like to avoid this approach as it requires very large numbers, moreover, assuming lexicographic preferences from the outset leads to more straightforward algorithms.     
		
To the best of our knowledge, matchings with prerequisite constraints have not been studied yet from an algorithmic perspective.  Some connections can be found in the literature on scheduling with precedence constraints \cite{LK78}, but, unlike in the scheduling domain, there is no common optimality criterion for all the agents, since their desires are often conflicting and all have to be taken into account.

We would however like to draw the reader's attention to the works of Guerin et al.\ \cite{GHFMG12} and Dodson et al.\ \cite{DMGG13}, who analyse a version of a course selection problem in greater depth, using probabilistic methods. Their work is a part of a larger research programme that involves advising college students about what courses to study and when, taking into account not only the required course prerequisites, but also the students' course histories and obtained grades. Based on this information, the authors try to estimate a student's ability to take multiple courses concurrently, with the goal to optimise the student's total expected utility and her chances of moving successfully toward graduation. Guerin et al.\ and  Dodson et al.\ consider also temporal factors, meaning that a student can only take a certain course in the current semester if she has completed the necessary prerequisites during previous semesters. By contrast, here we assume that students choose all their courses as well as their necessary prerequisites simultaneously, and we concentrate on computational problems connected with producing a matching that fulfils a global welfare criterion. 

\subsubsection*{Our contribution}
As mentioned above, we will formally introduce three possible models of course allocation involving prerequisites or corequisites. In the first case the prerequisites are antisymmetric and compulsory (i.e., a constraint might ensure that an applicant can subscribe to course $c_1$ only if she also subscribes to course $c_2$, but she can attend $c_2$ without attending $c_1$). For additive preferences we show that computing a Pareto optimal matching is an NP-hard problem.\footnotemark~In the case of lexicographic preferences we illustrate that the simple sequential allocation mechanism and its natural modification may output a matching that either violates prerequisites or Pareto optimality.  Therefore we stipulate that on her turn, an applicant chooses her most preferred course together with all the necessary prerequisites. Still, it is impossible to obtain each possible Pareto optimal matching in this way. It is also unlikely that an efficient algorithm will be able to produce all Pareto optimal matchings, since the problem of checking whether a given matching admits a Pareto improvement is NP-complete.
Considering structural properties of Pareto optimal matchings, it is known that finding a Pareto optimal matching with minimum cardinality is NP-hard, even in the very restricted one-to-one model (naturally without prerequisites) \cite{ACMM04}, but here we show that the problem of finding such a matching with maximum cardinality is also NP-hard.

The second model involving alternative prerequisites (i.e., where a constraint takes the form that an applicant can attend course $c_1$ only if she also attends either course $c_2$ or course $c_3$) seems to be computationally the most challenging case.  We show that although a Pareto optimal matching always exists, it cannot be computed efficiently unless P=NP, both under additive as well as lexicographic preferences of applicants.\footnotemark[\value{footnote}]\footnotetext{This follows because, as we will in fact show, the problem of finding a most-preferred feasible bundle of courses for a given applicant in this setting is NP-hard.}
 
For the third case with corequisites (i.e., an applicant can attend a course if and only if she attends all its corequisites) we propose another modification of the sequential allocation mechanism for finding Pareto optimal matchings. If the corequisites for all the applicants are the same, the model is closely related to matchings with sizes \cite{BM14} or many-to-many matchings with price-budget constraints \cite{CEFMMP14}, and we strengthen the existing results by showing that the problem of finding a maximum size Pareto optimal matching is not approximable within $N^{1-\varepsilon}$, for any $\varepsilon>0$, unless P=NP, where $N$ is the total capacity of the applicants.

The organisation of the paper is as follows. In Section \ref{s_def} we give  formal definitions of the problem models and define relevant notation and terminology. Sections \ref{s_pr}, \ref{s_alt} and  \ref{s_coreq} deal separately with the three different models, involving compulsory prerequisites, alternative prerequisites and corequisites respectively. Finally, Section \ref{s_open} concludes with some open problems and possible directions for future research.

\section{Definitions and notation}\label{s_def}
\subsection{Basic Course Allocation problem}
An instance of the \emph{Course Allocation problem} ({\sc ca}) involves a set $A=\{a_1,a_2,\dots,a_{n_1}\}$ of \emph{applicants} and a set $C=\{c_1,c_2,\dots,c_{n_2}\}$ of \emph{courses}.  Each course $c_j\in C$ has a capacity $q(c_j)$ that denotes the maximum number of applicants that can be assigned to $c_j$. Similarly each applicant $a_i\in A$ has a capacity $q(a_i)$ denoting the maximum number of courses that she can attend. The vector ${\mathbf q}$ denotes applicants' and courses' capacities. Moreover $a_i$ has a strict linear order (preference list) $P(a_i)$ over a subset of $C$.  We shall represent $a_i$'s preferences $P(a_i)$ as a simple ordered list of courses, from the most preferred to the least preferred one.  With some abuse of notation, we shall say that a course $c_j$ is \emph{acceptable} to applicant $a_i$ if $c_j\in P(a_i)$, otherwise $c_j$ is \emph{unacceptable} to $a_i$.  $\PP$ denotes the $n_1$-tuple of applicants' preferences. 
Thus altogether, the tuple $I=(A,C,{\mathbf q},\PP)$ constitutes an instance of {\sc ca}.
 
An \emph{assignment} $M$ is a subset of $A\times C$. 
The set of applicants assigned to a course $c_j\in C$ will be denoted by $M(c_j)=\{a_i\in A : (a_i,c_j)\in M\}$ and similarly, the \emph{bundle} of courses assigned to an applicant $a_i$ is $M(a_i)=\{c_j\in C : (a_i,c_j)\in M\}$. 
An assignment $M$ is a \emph{matching} if, for each $a_i\in A$, $M(a_i)\subseteq P(a_i)$ and $|M(a_i)|\leq q(a_i)$, and for each $c_j\in C$, $|M(c_j)|\le q(c_j)$.  In the presence of prerequisites and corequisites, additional \emph{feasibility} constraints are to be satisfied by a matching, which will be defined below.  An applicant $a_i\in A$ is said to be \emph{undersubscribed} (respectively \emph{full}) in a matching $M$ if $|M(a_i)|<q(a_i)$ (respectively $|M(a_i)|=q(a_i)$).  Similarly we may define \emph{undersubscribed} and \emph{full} for a course $c_j$ relative to $M$.

An applicant $a_i\in A$ has \emph{additive preferences} over bundles of courses if she has a utility $u_{a_i}(c_j)$ for each course $c_j\in C$, and she prefers a bundle of courses $C_1\subseteq C$ to another bundle $C_2\subseteq C$ if and only if $\sum_{c_j\in C_1}u_{a_i}(c_j)>\sum_{c_j\in C_2}u_{a_i}(c_j)$.

Applicant $a_i$ 
compares bundles of courses \emph{lexicographically} if, given two different bundles $C_1\subseteq C$ and $C_2\subseteq C$ she prefers $C_1$ to $C_2$ 
if and only if her most preferred course in the symmetric difference $C_1 \oplus C_2$ belongs to $C_1$.
Notice that the lexicographic ordering on bundles of courses generated by a strict preference order $P(a_i)$ is also strict. 

Applicant $a_i$ prefers matching $M'$ to matching $M$ if she prefers $M'(a_i)$ to $M(a_i)$. 
We say that a matching $M'$ \emph{(Pareto) dominates} a matching $M$ 
if  at least one applicant prefers $M'$ to $M$ and no applicant prefers $M$ to $M'$. 

A \emph{Pareto optimal matching} (or \emph{POM} for short), is a matching that is  not (Pareto) dominated by any other matching. As the dominance relation 
is a partial order over $\MM$, the set of all matchings in $I$, and $\MM$ is finite, a POM exists for each instance of {\sc ca}. 

\subsection{Compulsory prerequisites}
\label{sec:comp-prereq}
We now define the first extension of {\sc ca} involving compulsory prerequisites.  Suppose that for each applicant $a_i\in A$, there is a strict partial order $\to_{a_i}$ on the set of courses $C$ 
representing the prerequisites of applicant $a_i$.  It is easy to see that this partial order can be fully specified if for each $c_j\in C$ its \emph{immediate prerequisites} for $a_i$ are given, where $c_j'\in C$ is an immediate prerequisite of $c_j$ if $c_j\to_{a_i} c_j'$ and there is no $c_j''\in C$ such that $c_j\to_{a_i} c_j''\to_{a_i} c_j'$ (i.e., $\to_{a_i}$ is given in terms of its transitive reduction).  
We now define the \emph{feasibility} of a bundle of courses relative to the constraints on compulsory prerequisites. 
\begin{defn}
A bundle of courses $C'\subseteq C$ is \emph{feasible} for an applicant $a_i\in A$ if the following three conditions are fulfilled: 
\label{def:compulsory}
\begin{itemize}\itemsep0pt
\item[$(i)$] $C'\subseteq P(a_i)$ ;
\item[$(ii)$] $|C'|\leq q(a_i)$ ;
\item[$(iii)$] $C'$ fulfils $a_i$'s prerequisites, i.e., for each $c_j,c_k\in C$, if $c_j\in C'$ and $c_j\to_{a_i} c_k$ then $c_k\in C'$.
\end{itemize}
\end{defn}

A subset $C'$ of a partially-ordered set $C$ fulfilling condition $(iii)$ is called a \emph{down-set} (see \cite{DP90}). We shall denote by $\stackrel{\rightarrow_{a_i}}{c_j}$ the inclusion minimal down-set of $C$ (with respect to $a_i$'s prerequisites) that contains course $c_j$.

For technical reasons, we assume that $\stackrel{\rightarrow_{a_i}}{c_j}\subseteq P(a_i)$ for each $a_i\in A$  and each $c_j\in P(a_i)$. If this is not the case then we can easily modify the preference list of applicant $a_i$ either by deleting a course $c_j$ if $P(a_i)$ does not contain all the courses in $\stackrel{\rightarrow_{a_i}}{c_j}$, or we can append the missing courses to the end of $P(a_i)$.

An instance $I$ of the \emph{Course Allocation problem with (compulsory) PRequisites} ({\sc capr}) comprises a tuple $I=(A,C,{\mathbf q},\PP,\rightarrow)$, where $(A,C,{\mathbf q},\PP)$ is an instance of {\sc ca} and $\rightarrow$ is the $n_1$-tuple of prerequisite partial orders $\rightarrow_{a_i}$ for each applicant $a_i\in A$.  In an instance of {\sc capr}, a \emph{matching} $M$ is as defined in the {\sc ca} case, together with the additional property that, for each applicant $a_i\in A$, $M(a_i)$ is feasible for $a_i$.

In an instance of {\sc capr} where the prerequisites are the same for all applicants, we may drop the applicant subscript when referring to the prerequisite partial order $\rightarrow$.

\subsection{Alternative prerequisites}
The second model is a variant of {\sc capr} in which prerequisites need not be compulsory but are in general presented in the form of alternatives.  Formally, each applicant $a_i$ has a mapping $\mapsto_{a_i}: C\to 2^C$ with the following meaning: if $c_j\mapsto_{a_i} \{c_{i_1},c_{i_2},\dots,c_{i_k}\}$, it must then hold that if $a_i$ wants to attend course $c_j$ then she has to attend at least one of the courses $c_{i_1},c_{i_2},\dots,c_{i_k}$ too.  For each applicant $a_i\in A$ the relation $\mapsto_{a_i}$ must be acyclic in the sense that the graph $G_{a_i}=(V,E_{a_i})$ must be acyclic, where $V=C$ and $E_{a_i}=\{(c_j,c_k)\in C\times C : c_j\mapsto_{a_i} C'\wedge c_k\in C'\}$.

We thus define a bundle of courses $C'\subseteq C$ to be \emph{feasible} for a given applicant $a_i\in A$ if Conditions (i) and (ii) in Definition \ref{def:compulsory} are satisfied, and moreover, for any course $c_j\in C'$, if $c_j\mapsto_{a_i} \{c_{i_1},c_{i_2},\dots,c_{i_k}\}$, then $c_{i_r}\in C'$ for some $r$ ($1\leq r\leq k$).
 
An instance $I$ of the \emph{Course Allocation problem with Alternative PRequisites} ({\sc caapr}) comprises a tuple $I=(A,C,{\mathbf q},\PP,\mapsto)$, where $(A,C,{\mathbf q},\PP)$ is an instance of {\sc ca} and $\mapsto$ is the $n_1$-tuple of mappings  $\mapsto_{a_i}$ for each applicant $a_i\in A$.  In an instance of {\sc caapr}, a \emph{matching} $M$ is as defined in the {\sc ca} case, together with the additional property that, for each applicant $a_i\in A$, $M(a_i)$ is feasible for $a_i$.

\subsection{Corequisites}
\label{sec:coreq}
In the third and final model we assume that constraints on courses are given in the form of corequisites.  We assume that corequisite constraints are not applicant-specific, and hence there is a single reflexive, symmetric and transitive relation $\leftrightarrow$ on $C$ such that each applicant is allowed to subscribe to a course $c_j\in C$ only if she also subscribes to each $c_k\in C$ such that  $c_j\leftrightarrow c_k$.  Two courses $c_j,c_k\in C$ with $c_j\leftrightarrow c_k$ are said to be each other's \emph{corequisites}.  Relation $\leftrightarrow$ is an equivalence relation on $C$ and it   
effectively partitions the set of courses into equivalence classes $C^1, C^2,\dots, C^{r}$. Hence an applicant can subscribe either to all courses in one equivalence class or to none.\footnote{This corresponds to the case where each set of courses in $C^i$ can effectively be replaced by a single \emph{supercourse} taking up $|C_i|$ places of any applicant it is assigned to -- we shall exploit this correspondence in Section \ref{sec:coreq}.}

Formally, we define a bundle of courses $C'\subseteq C$ to be \emph{feasible} for a given applicant $a_i\in A$ if Conditions (i) and (ii) in Definition \ref{def:compulsory} are satisfied, and moreover, for any two courses $c_j,c_k\in C$, if $c_j\leftrightarrow c_k$ then $c_j\in C'$ if and only if $c_k\in C'$.  An instance $I$ of the \emph{Course Allocation problem with Corequisites} ({\sc cacr}) comprises a tuple $I=(A,C,{\mathbf q},\PP,\leftrightarrow)$, where $(A,C,{\mathbf q},\PP)$ is an instance of {\sc ca} and $\leftrightarrow$ is as defined above.  In an instance of {\sc cacr}, a \emph{matching} $M$ is as defined in the {\sc ca} case, together with the additional property that, for each applicant $a_i\in A$, $M(a_i)$ is feasible for $a_i$.

We remark that we do not consider corequisites as a special case of compulsory prerequisites, nor vice versa,  for in the definition compulsory prerequisites,  we stipulate that the order relation is antisymmetric, while for corequisites symmetry is required.

\section{Compulsory prerequisites}\label{s_pr}
In the presence of compulsory prerequisites, we consider the case of additive preferences in Section \ref{sec:comp-add} and lexicographic preferences in Section \ref{sec:comp-lex}.  It turns out that the problem of finding a POM under additive preferences is NP-hard, as we show in Section \ref{sec:comp-add}.  Thus the majority of our attention is focused on the case of lexicographic preferences in Section \ref{sec:comp-lex}.  In that section we mainly consider algorithmic questions associated with the problems of (i) finding a POM, (ii) testing a matching for Pareto optimality, and (iii) finding a POM of maximum size.

\subsection{Additive preferences}
\label{sec:comp-add}
First we show that the assumption of additive preferences in {\sc capr} makes it difficult to compute a POM.

\begin{lem}\label{t_add}
The problem of finding a most-preferred feasible bundle of courses of a given applicant with additive preferences in {\sc capr} is NP-hard.
\end{lem}
\begin{proof} We transform from the {\sc knapsack} problem, which is defined as follows. An instance $I$ of {\sc knapsack} comprises a set of integers $w_1,w_2,\dots,w_n,p_1,p_2,\dots, p_n$, $W, P$.  The problem is to decide whether there exists a set $K\subseteq \{1,2\dots,n\}$ such that $\sum_{i\in K} w_i\le W$ and $\sum_{i\in K} p_i\ge P$.  {\sc knapsack} is NP-complete \cite{Kar72}.

Let us construct an instance $J$ of {\sc capr} as follows.  There is a single applicant $a_1$ such that $q(a_1)=W$.  The set of courses is $C\cup D$, where $C=\{c_1,c_2,\dots,c_n\}$ and $D=\cup_{i=1}^n \{d_i^1,d_i^2,\dots,d_i^{w_i-1}\}$.  For each course $c_i$ $(1\leq i\leq n)$, its utility for applicant $a_1$ is equal to $p_i+\delta_i$, where $\delta_i$ will be defined shortly.  Moreover $c_i$ has $w_i-1$ prerequisites $d_i^1,d_i^2,\dots, d_i^{w_i-1}$.
For each $i$ ($1\leq i\leq n)$ and $j$ ($1\leq j\leq w_i-1$), $d_i^j$ has utility $\varepsilon_i^j$ for $a_1$, where the $\delta$ and $\varepsilon$ values are all positive, distinct and add up to less than 1.  They are selected simply to ensure that $a_1$'s preferences over individual courses are strict.
It is easy to see that in $J$, $a_1$ has a feasible bundle with utility at least $P$ if and only if $I$ is a yes-instance of {\sc knapsack}. 

Clearly a polynomial-time algorithm for finding a most-preferred feasible bundle of courses for $a_1$ can be used to determine whether $a_1$ has a feasible bundle in $J$ with utility at least $P$, hence the result.
\end{proof}

In the case of just one applicant $a_1$, a matching $M$ is a POM if and only if $a_1$ is assigned in $M$ a most-preferred feasible bundle of courses, otherwise $M$ is dominated by assigning $a_1$ to this bundle.  Hence the above lemma directly implies the following result.

\begin{thm}
\label{thm:findPOMcapr}
Given an instance of {\sc capr} with additive preferences, the problem of finding a POM is NP-hard.
\end{thm}

Given the above negative result, we do not pursue additive preferences any further in this section, and instead turn our attention to lexicographic preferences.

\subsection{Lexicographic preferences}
\label{sec:comp-lex}
\subsubsection{Finding a POM}
In this section we explore variants of the sequential allocation mechanism, and show that one formulation allows us to find a POM in polynomial time.  This mechanism, referred to as {\sf SM-CAPR}, does however have some drawbacks: it is not truthful (see Section \ref{s_open}) and it is not able to compute all POMs in general.

In the context of course allocation when there are some dependencies among courses (for instance the constraints on prerequisites in {\sc capr}) the standard sequential mechanism might output an allocation that does not fulfil some constraints on prerequisites. On the other hand, if we require that an applicant can only choose a course if she is already assigned all its prerequisites, the output may be a matching that is not Pareto optimal.  This is illustrated by the following example.

\begin{example}\label{ex1}
Construct a {\sc capr} instance in which $A=\{a_1,a_2\}$ and $C=\{c_1,c_2,c_3,c_4\}$.  Each applicant has capacity $2$ and  each course has capacity  $1$.  The prerequisites of both applicants are the same, and are as follows:
\[
c_1\to c_3; \qquad c_2\to c_4.
\]
The applicants have the following preference lists:
\begin{eqnarray*}
P(a_1): c_1, c_2, c_4, c_3\\
P(a_2): c_2, c_1, c_3, c_4
\end{eqnarray*}


The sequential allocation mechanism with policy $\sigma=a_1,a_2,a_1,a_2$ will assign to applicant $a_1$ the bundle $\{c_1,c_4\}$ and to applicant $a_2$ the bundle  $\{c_2,c_3\}$. Clearly, neither of the assigned bundles fulfils the prerequisites.

Now suppose that in the sequential allocation mechanism an applicant is allowed to choose the most-preferred undersubscribed course \emph{for which she already has all the prerequisites}. Let the policy start with $a_1,a_2$. Applicant $a_1$ can choose neither $c_1$ nor $c_2$, as these courses require a prerequisite that she is not assigned yet. So she chooses $c_4$. Similarly, applicant $a_2$ will choose $c_3$. When these applicants are allowed to pick their next course, irrespective of the remainder of the policy, $a_1$ must choose $c_2$ and $a_2$ must choose $c_1$. So in the resulting matching $M$ we have $M(a_1)=\{c_2,c_4\}$ and $M(a_2)=\{c_1,c_3\}$. This matching is clearly not Pareto optimal, since both applicants (having lexicographic preferences) will strictly improve by exchanging their assignments. \qed
\end{example}

Therefore we propose a variant of the sequential allocation mechanism, denoted by {\sf SM-CAPR}, that can be regarded as being ``in between'' the serial dictatorship mechanism and the general sequential allocation mechanism. Suppose a policy $\sigma$ is fixed; again one applicant can appear in $\sigma$ several times, up to her capacity.  Applicant $a_i$ on her turn identifies her most-preferred course $c_j$ that she has not yet considered, and that she is not already assigned to (if no such course $c_j$ exists then $a_i$ is said to have \emph{exhausted} her preference list and will be assigned no more courses).  If $c_j$ is full or $a_i$ is already assigned to $c_j$ then $a_i$ considers the next course on her list (on the same turn).  We then compute the down-set  $\stackrel{\rightarrow_{a_i}}{c_j}$ of $c_j$ (we shall explain how to do this efficiently in the proof of Theorem \ref{thm_alg} below).  If all courses $c_k$ in $\stackrel{\rightarrow_{a_i}}{c_j}$ satisfy the property that either (i) $c_k$ has a free place or (ii) $c_k$ is already assigned to $a_i$, and the number of courses in $\stackrel{\rightarrow_{a_i}}{c_j}$ not already assigned to $a_i$ does not exceed the remaining capacity of $a_i$, then $a_i$ is assigned the bundle $\stackrel{\rightarrow_{a_i}}{c_j}$.
If it is impossible to assign bundle $\stackrel{\rightarrow_{a_i}}{c_j}$ to $a_i$ (as signified by the boolean \emph{feasible} becoming {\tt false}) then $a_i$ moves to the next course in her preference list until either she is assigned to some bundle or her preference list is exhausted.  This completes a single turn for $a_i$.

The pseudocode for {\sf SM-CAPR} is given in Algorithm \ref{alg:smcapr}.  This algorithm constructs a POM $M$ in a given {\sc capr} instance $I$ relative to a given policy $\sigma$.  Notice that serial dictatorship will be obtained as a special case of {\sf SM-CAPR} if all the copies of one applicant form a substring (i.e., a contiguous subsequence) of the policy.
\begin{algorithm}[t]
\small
\caption{{\sf SM-CAPR}}\label{alg:smcapr}
\begin{algorithmic}[1]
\Require {\sc capr} instance $I$ and a policy $\sigma$
\Ensure return $M$, a POM in $I$
\State $M := \emptyset$;
\For{{\bf each} applicant $a_i\in \sigma$ in turn} \label{line:forloop}
  \State feasible $:=$ {\tt false};
	\While{$a_i$ has not exhausted her preference list {\bf and not} feasible} \label{line:whileloop}
		\State $c_j :=$ next course in $a_i$'s list; \label{line:next}
		\If{$c_j\notin M(a_i)$ and $c_j$ is undersubscribed} \label{line:under}
			\State $S := \stackrel{\rightarrow_{a_i}}{c_j}$; \label{line:downset}
			\State feasible := {\tt true};
			\For{{\bf each} $c_k\in S$} \label{line:loop1start}
				\If{$c_k\in M(a_i)$}
					\State $S := S\backslash \{c_k\}$;
				\ElsIf{$c_k$ is full} \label{line:full}
					\State feasible := {\tt false};
				\EndIf
			\EndFor \label{line:loop1end}
			\If{feasible}
				\If{$|M(a_i)|+|S|\leq q(a_i)$} \label{line:quota}
					\For{{\bf each} $c_k\in S$} \label{line:loop2start}
						\State $M:=M\cup \{(a_i,c_k)\}$;
					\EndFor \label{line:loop2end}
				\Else
					\State feasible $:=$ {\tt false};
				\EndIf
			\EndIf
		\EndIf
	\EndWhile
\EndFor \label{line:penultimate}
\State \Return $M$; \label{line:return}
\end{algorithmic}
\normalsize
\end{algorithm}
We now show that {\sf SM-CAPR} constructs a POM and is an efficient algorithm.

\begin{thm}\label{thm_alg}
Algorithm {\sf SM-CAPR} produces a POM for a given instance $I$ of {\sc capr} and for a given policy $\sigma$ in $I$.  The complexity of the algorithm is $O(N+n_2(L+\Delta))$, where $N$ is the sum of the applicants' capacities, $n_2$ is the number of courses, $L$ is the total length of the applicants' preference lists and $\Delta$ is the total number of immediate prerequisites of each course $c_j$ in $\to_{a_i}$, taken over each applicant $a_i$.
\end{thm} 
\begin{proof} It is straightforward to verify that the assignment $M$ produced by {\sf SM-CAPR} is a matching in $I$.  Suppose for a contradiction that $M$ is not a POM.  Then there exists a matching $M'$ that dominates $M$. Let $A'$ be the set of applicants who prefer their assignment in $M'$ to their assignment in $M$. 
Define a \emph{stage} to be an iteration of the {\bf while} loop, and for a given stage $s$, define its \emph{number} to be the integer $k$ such that $s$ is the $k$th iteration of the {\bf while} loop, taken over the entire execution of the algorithm.  For each $a_j\in A'$, consider the first stage where a course, say $c_{i_j}$, was identified for $a_j$ at line \ref{line:next}, such that $c_{i_j}\in M'(a_j)\backslash M(a_j)$, and let $s_j$ be the number of this stage.  Let $a_k=\arg\min_{a_j\in A'} \{s_j\}$.

As $c_{i_k}\in M'(a_k)$, all the prerequisites of $c_{i_k}$ also belong to $M'(a_k)$.  Since $s_k$ is the first stage in which a course $c_{i_k}$ was identified for $a_k$ in line \ref{line:next}, such that $c_{i_k}\in M'(a_k)\backslash M(a_k)$, all the courses assigned in $M$ to any applicant $a_j$ in previous stages also belong to $M'(a_j)$, for otherwise $M'$ does not dominate $M$.  Also, clearly $c_{i_k}\notin M(a_k)$.  Thus it was not the case that applicant $a_k$ failed to receive course $c_{i_k}$ in $M$ at stage $s_k$ because $a_k$ did not have room for $c_{i_k}$ and all of its prerequisites not already assigned to her in $M$.  Rather, applicant $a_k$ failed to receive course $c_{i_k}$ in $M$ at stage $s_k$ because at least one course in $\stackrel{\rightarrow_{a_k}}{c_{i_k}}$, say $c_r$, was already full in $M$ before stage $s_k$.  It follows from our previous remark that in $M'$, all the places in $c_r$ are occupied by applicants other than $a_k$.  Thus $c_{i_k}$ cannot be assigned to $a_k$ in $M'$ after all, a contradiction.

To derive the complexity bound of the algorithm, let us first consider the representation of the strict partial order $\to_{a_i}$ of prerequisites for each applicant $a_i$.  As mentioned in Section \ref{sec:comp-prereq}, we can assume that the transitive reduction of $\to_{a_i}$, representing the immediate prerequisites of each course, is given as an input to the algorithm.  Furthermore, we can assume that this is represented as a directed acyclic graph (DAG) $D_{a_i}$ in adjacency list form.  The first task is to construct the transitive closure of $D_{a_i}$ for each $a_i\in A$, giving a DAG $D^*_{a_i}$, again in adjacency list form, from which $\stackrel{\rightarrow_{a_i}}{c_j}$, required in line \ref{line:downset}, can be returned as a list in $O(n_2)$ time and space, for any course $c_j\in C$.  The transitive closure of $D_{a_i}$ can be found in $O(n_2 \Delta_{a_i})$ time, where $\Delta_{a_i}$ is the number of arcs in $D_{a_i}$, using standard graph traversal techniques, meaning that the total time required to compute $D^*_{a_i}$ for all $a_i\in A$ is $O(n_2\Delta)$, where $\Delta=\sum_{a_i\in A} \Delta_{a_i}$.

We next observe that the number of iterations of the {\bf for} loop in line \ref{line:forloop} is $O(|\sigma|=N)$, where $N$ is the sum of the applicants' capacities, whilst the total number of iterations of the {\bf while} loop in line \ref{line:whileloop}, taken over the entire execution of the algorithm, is $O(L)$, where $L$ is the total length of the applicants' preference lists.  At line \ref{line:next} we assume that the next course for an applicant $a_i$ is maintained by a pointer that initially points to the head of $a_i$'s preference list, and once the course $c_j$ pointed to by $a_i$'s pointer has been found, the pointer moves on one position (if $c_j$ was the last course on $a_i$'s list then the pointer becomes null, indicating that $a_i$ has exhausted her list).  Each of the for loops occupying lines \ref{line:loop1start}-\ref{line:loop1end} and \ref{line:loop2start}-\ref{line:loop2end} has $O(n_2)$ overall complexity.  Thus the overall complexity of the algorithm is $O(N+n_2(L+\Delta))$ as claimed.
\end{proof}

The complexity of {\sf SM-CAPR} is no better than $O(N+Ln_2)$ in the worst case, as the following example shows.

\begin{example}
Consider a {\sc cacr} instance in which $A=\{a_1,a_2,\dots,a_n\}$ is the set of applicants and $C=\{c_1,c_2,\dots,c_{2n}\}$ is the set of courses, for some $n\geq 1$. Assume that each course has capacity 1, whilst each applicant has capacity $n$ and ranks all courses in increasing indicial order.  Also suppose that the prerequisites for each applicant are as follows:
$$c_1 \to c_2\to \dots \to c_{2n}.$$
There are $n$ POMs: in the POM $M_i$ ($1\leq i\leq n$), $a_i$ is assigned the set of courses $\{c_{n+1},c_{n+2},\dots,c_{2n}\}$ and no course is assigned to any other applicant.  Given any policy $\sigma$, let $a_i$ be the first applicant considered during an execution of {\sf SM-CAPR}.  When $\stackrel{\rightarrow_{a_i}}{c_1}$ is constructed in line \ref{line:downset}, $2n$ courses are returned; likewise when $\stackrel{\rightarrow_{a_i}}{c_2}$ is constructed, $2n-1$ courses are returned, and so on.  This continues until $\stackrel{\rightarrow_{a_i}}{c_{n+1}}$ is constructed, leading to the matching $M_i$ being formed at this while loop iteration.  Note that even if the process of constructing $S=\stackrel{\rightarrow_{a_i}}{c_j}$ were to halt as soon as $|S|>q(a_i)$, the total number of courses checked at this step of {\sf SM-CAPR} is still $\Omega(n^2)$.  Similarly, the number of courses checked at each other applicant's turn in the policy is also $\Omega(n^2)$; the only difference being that in each such case {\sf SM-CAPR} determines that $c_{n+i}$ is full, for each $i$ ($1\leq i\leq n)$.  The overall number of steps used by {\sf SM-CAPR} is then $\Omega(n^3)=\Omega(N+Ln_2)$. \qed
\end{example}

Our next example indicates that in general, {\sf SM-CAPR} is not capable of generating all POMs for a given {\sc capr} instance.


\begin{example}\label{ex3}
{\sf SM-CAPR} is not able to produce all POMs, even in the case when there are only two applicants $a_1,a_2$ and the capacity of each course is 1. We provide two instances to illustrate this. In  $I_1$ the prerequisites of all applicants are the same.  In $I_2$ they are different, but each course has at most one prerequisite.

In $I_1$, we have  $C=\{c_1,c_2,c_3\}$, and course $c_1$ has two prerequisites as follows: 
\begin{equation}\label{pq2}
c_1\to c_2; \qquad c_1\to c_3.
\end{equation}
Each applicant has capacity 3 and the following preference list: $c_1,c_2,c_3$.

Depending on the policy, {\sf SM-CAPR} outputs either the matching that assigns all three courses to $a_1$, or  the matching that assigns all three courses to $a_2$. However, it is easy to see that the two matchings that assign $c_2$ to one applicant and  $c_3$ to the other one are also Pareto optimal.

In $I_2$, we have  $C=\{c_1,c_2,c_3\}$. Now the prerequisites of the applicants are different:
\begin{equation}\label{pq3}
c_1\to_{a_1} c_3; \qquad c_2\to_{a_2} c_3,
\end{equation}
Each applicant has capacity $2$ and their preferences are as follows:
$$
P(a_1): c_1, c_2, c_3\qquad\
P(a_2): c_2, c_1, c_3. \\
$$
There are $4$ different POMs, as follows:
\[
\begin{array}{ll}
M_1(a_1)=\{c_1,c_3\}, & M_1(a_2)=\emptyset; \\ 
M_2(a_1)=\emptyset, &  M_2(a_2)=\{c_2,c_3\}; \\
M_3(a_1)=\{c_2\}, & M_3(a_2)=\{c_1,c_3\};\\
M_4(a_1)=\{c_2,c_3\}, & M_3(a_2)=\{c_1\}.
\end{array}
\]

{\sf SM-CAPR} outputs $M_1$ with a policy in which $a_1$ is first, and $M_2$ with a policy in which $a_2$ is first.  Notice that neither $M_3$ nor $M_4$ can be obtained by {\sf SM-CAPR}. \qed
\end{example}

Theorem \ref{thm:findPOMcapr} shows that finding a POM in the presence of additive preferences is NP-hard.  It is instructive to show where {\sf SM-CAPR} can fail to find a POM in this context,
\begin{example}
Let $I$ be an instance of {\sc capr} in which there are two applicants, $a_1,a_2$, each of which has capacity 2, and four courses, $c_1,c_2,c_3,c_4$, each of which has capacity 1.  The prerequisites of both applicants are the same, and are as follows:
\[
c_1\to c_2; \qquad c_3\to c_4.
\]
The utilities of the courses for the applicants are as follows:
\begin{eqnarray*}
u_{a_1}(c_1)=u_{a_2}(c_3)=6\\
u_{a_1}(c_3)=u_{a_2}(c_1)=4\\ 
u_{a_1}(c_4)=u_{a_2}(c_2)=3\\
u_{a_1}(c_2)=u_{a_2}(c_4)=0
\end{eqnarray*}
Regardless of the policy, {\sf SM-CAPR} constructs the matching $M=\{(a_1,c_1),(a_1,c_2),(a_2,c_3),$ $(a_2,c_4)\}$.  $M$ is not a POM as it is dominated by $M'=\{(a_1,c_3),(a_1,c_4),(a_2,c_1),(a_2,c_2)\}$. \qed
\end{example}

\subsubsection{Testing for Pareto optimality}
In the previous subsection we gave a polynomial-time algorithm for constructing a POM in an instance of {\sc capr}.  It is also reasonable to expect that an alternative approach could involve starting with an arbitrary matching, and for as long as the current matching $M$ is dominated, replace $M$ by any matching that dominates it.  However, the difficulty with this method is that the problem of determining whether a matching is Pareto optimal is computationally hard, as we demonstrate by our next result.  This hardness result also shows that there is unlikely to be a ``nice'' (polynomial-time checkable) characterisation of a POM, in contrast to the case where there are no prerequisites \cite{CEFMMP14}.  We firstly define the following problems:\\

\begin{minipage}{5.5in}
\emph{Name:} \xc\\
\emph{Instance:} a set $X=\{x_1,x_2,\dots, x_{3n}\}$ and a set $\mathcal T=\{T_1,T_2,\dots,T_m\}$ such that for each $i$ ($1\leq i\leq m$), $T_i\subseteq X$ and $|T_i|=3$.\\
\emph{Question:} is there a subset $\mathcal T'$ of $\mathcal T$ such that $T_i\cap T_j=\emptyset$ for each $T_i,T_j\in \mathcal T'$ and $\cup_{T_i\in \mathcal T'} T_i=X$?\\
\end{minipage}

\begin{minipage}{5.5in}
\emph{Name:} {\sc dom capr}\\
\emph{Instance:} an instance $I$ of {\sc capr} and a matching $M$ in $I$\\
\emph{Question:} is there a matching $M'$ that dominates $M$ in $I$?\\
\end{minipage}

\begin{thm}\label{t_POM} 
{\sc dom capr} is NP-complete even if each course has capacity 1 and has at most one immediate prerequisite for each applicant.
\end{thm}
\begin{proof}
Clearly {\sc dom capr} is in NP.  To show NP-hardness, we reduce from \xc, which is NP-complete \cite{Kar72}.  Let $I$ be an instance of \xc\ as defined above.  For each $T_i\in \mathcal T$, let us denote the elements that belong to $T_i$ by $x_{i_1},x_{i_2},x_{i_3}$. Obviously, we lose no generality by assuming that $m\geq n$.

We construct an instance $J$ of {\sc dom capr} based on $I$ in the following way. 
The set of applicants is $A=\{a_1,a_2,\dots,a_{m+1}\}$. The capacities are $q(a_i)=4$ ($1\leq i \leq m$) and $q(a_{m+1})=2n+m$. The set of courses is $C=\mathcal T\cup X\cup Y\cup W$, where
$\mathcal T=\{T_1,T_2,\dots,T_m\}, X=\{x_1,x_2,\dots,x_{3n}\}$, $Y=\{y_1,y_2,\dots,y_m\}$ and $W=\{w_1,w_2,\dots,w_{m-n}\}$.
(Some of the courses in $J$ derived from the elements and sets in $I$ are denoted by identical symbols, but no ambiguity should arise.)  Each course has capacity 1.  The preferences of the applicants are:
\[
\begin{array}{rll}
P(a_i): & T_i,[W],y_i,x_{i_1},x_{i_2},x_{i_3} & (1\leq i\leq m)\\
P(a_{m+1}): & y_1,[X],[W],y_2,\dots, y_m
 \end{array}
\]
where the symbols $[W]$ and $[X]$ denote all the courses in $W$ and $X$, respectively, in an arbitrary strict order.  Recall that $\{x_{i_1},x_{i_2},x_{i_3}\}\subseteq X$ ($1\leq i\leq m$).  The prerequisites of applicants are:
\[
\begin{array}{rll}
a_i: &T_i\to_{a_i} x_{i_1}\to_{a_i} x_{i_2}\to_{a_i} x_{i_3} & (1\leq i \leq m)\\
a_{m+1}:& y_1\to_{a_{m+1}} y_2\to_{a_{m+1}} \dots \to_{a_{m+1}} y_m&
\end{array}
\]
Define the following matching:
\[M=\{(a_i,y_i) : 1\leq i\leq m\}\cup
\{(a_{m+1},x_j) : 1\leq j\leq 3n)\}\cup \{(a_{m+1},w_k) : 1\leq k\leq m-n)\}.
\]
We claim that $I$ admits an exact cover if and only if $M$ is dominated in $J$.

For, suppose that $\{T_{j_1},T_{j_2},\dots,T_{j_n}\}$ is an exact cover in $I$.  We construct a matching $M'$ in $J$ as follows.  For each $k$ ($1\leq k\leq n$), in $M'$, assign $a_{j_k}$ to $T_{j_k}$ and to $a_{j_k}$'s three prerequisites of $T_{j_k}$ that belong to $X$.  Let $A'=\{a_{j_1},a_{j_2},\dots,a_{j_n},a_{m+1}\}$ and let $A\backslash A'=\{a_{k_1},a_{k_2},\dots,a_{k_{m-n}}\}$.  For each $r$ ($1\leq r\leq m-n$), in $M'$, assign $a_{k_r}$ to $w_r$.  Finally in $M'$, assign $a_{m+1}$ to every course in $Y$.  It is straightforward to verify that $M'$ dominates $M$ in $J$.

Conversely, suppose there exists a matching $M'$ that dominates $M$ in $J$. Then, at least one applicant must be better off in $M'$ compared to $M$.
 
If $a_{m+1}$ improves, she must obtain $y_1$ and so, due to her prerequisites, all the courses in $Y$. This means that each applicant $a_i$ ($1\leq i\leq m$) must obtain a course that she prefers to $y_i$.

Each such applicant $a_i$ can improve relative to $M$ in two ways. Either she obtains in $M'$ a course in $W$, or she obtains $T_i$.  In the latter case then she must receive the corresponding courses $x_{i_1},x_{i_2},x_{i_3}$ in $M'$. In either of these two cases, since $a_{m+1}$ cannot be worse off, she must obtain in $M'$ the course $y_1$ and hence all courses in $Y$.

This means that in $M'$ all the applicants must strictly improve compared to $M$. As there are only $n-m$ courses in $W$, there are exactly $n$ applicant in $A\backslash \{a_{m+1}\}$ -- let these applicants be $\{a_{j_1},a_{j_2},\dots,a_{j_n}\}$ -- who obtain a course in $\mathcal T$ and its three prerequisites in $X$. As the capacity of each course is 1, it follows that $\{T_{j_k} : 1\leq j\leq n\}$ is an exact cover in $I$.
\end{proof}

We remark that the variant of {\sc dom capr} for additive preferences is also NP-complete by Theorem \ref{t_POM}, since lexicographic preferences can be viewed as a special case of additive preferences by creating utilities that steeply decrease in line with applicants' preferences -- see \cite{BL11} for more details.
\subsubsection{Finding large POMs}  
Example \ref{ex3} shows that an instance of {\sc capr} may admit POMs of different cardinalities, where the \emph{cardinality} of a POM refers to the number of occupied course slots. It is known that finding a POM with minimum cardinality is an NP-hard problem even for {\sc ha},  the House Allocation problem (i.e., the restriction of {\sc ca} in which each applicant and each course has capacity 1) \cite[Theorem 2]{ACMM04}. However, by contrast to the case for {\sc ha} \cite[Theorem 1]{ACMM04} and {\sc ca} \cite[Theorem 6]{CEFMMP14}), the problem of finding a maximum cardinality POM in the {\sc capr} context is NP-hard, as we demonstrate next via two different proofs.  Our first proof of this result shows that hardness holds even if the matching is not required to be Pareto optimal.

We firstly define some problems.  Let {\sc max capr} and {\sc max pom capr} denote the problems of finding a maximum cardinality matching and a maximum cardinality POM respectively, given an instance of {\sc capr}.  Let {\sc max capr-d} denote the decision version of {\sc max capr}: given an instance $I$ of {\sc capr} and an integer $K$, decide whether $I$ admits a matching of cardinality at least $K$.  We obtain {\sc max pom capr-d} from {\sc max pom capr} similarly.

\begin{thm}\label{t_MAX-CAPR-hard}
{\sc max capr-d} is NP-complete, even if each applicant has capacity 4 and each course has capacity 1, and the prerequisites are the same for all applicants.
\end{thm}
\begin{proof}
Clearly {\sc max capr-d} is in NP.  To show NP-hardness, we reduce from {\sc ind set-d} in cubic graphs; here {\sc ind set-d} is the decision version of {\sc ind set}, the problem of finding a maximum independent set in a given graph.  {\sc ind set-d} is NP-complete in cubic graphs \cite{GJS76,MS77}. Let $\langle G,K\rangle$ be an instance of {\sc ind set-d} in cubic graphs, where $G=(V,E)$ is a cubic graph and $K$ is a positive integer.  Assume that $V=\{v_1,v_2,\dots,v_n\}$ and $E=\{e_1,e_2,\dots,e_m\}$.  For a given vertex $v_i\in V$, let $E_i\subseteq E$ denote the set of edges incident to $v_i$ in $G$.  Clearly $|E_i|=3$ as $G$ is cubic.

We form an instance $I$ of {\sc max capr-d} as follows.  Let $A$ be the set of applicants and let $V\cup E$ be the set of courses, where $A=\{a_i : v_i\in V\}$ (we use the same notation for the vertices and edges in $G$ as we do for the courses in $I$, but no ambiguity should arise.)  
Let the capacity of each applicant be 4 and the capacity of each course be 1.
The preference list of each applicant is as follows:
\[a_i : v_i ~~ [E_i] ~~~~ ~~~~ (1\leq i\leq n)\]
where the symbol $[E_i]$ denotes all members of $E_i$ listed in arbitrary order.  
For each $v_i\in V$ and for each $e_j\in E_i$, define the prerequisite $v_i\rightarrow e_j$ for all applicants.  We claim that $G$ has an independent set of size at least $K$ if and only if $I$ has a matching of size at least $m+K$.

For, suppose that $S$ is an independent set in $G$ where $|S|\geq K$.  Let $A'=\{a_i\in A : v_i\in S\}$.  We form an assignment $M$ in $I$ as follows.  For each applicant $a_i\in A'$, assign $a_i$ to $v_i$ plus the prerequisite courses in $E_i$.  Then for each applicant $a_i\notin A'$, assign $a_i$ to any remaining courses in $E_i$ (if any).  It is straightforward to verify that $M$ is a matching in $I$.  Also 
$|M|=m+|S|\geq m+K$, since every applicant $a_i\in A'$ obtains $v_i$ and all prerequisite courses in $E_i$, and then the applicants in $A\backslash A'$ are collectively assigned to all remaining unmatched courses in $E$.

Conversely suppose that $M$ is a matching in $I$ such that $|M|\geq m+K$.
Let $S$ denote the courses in $V$ that are matched in $M$, and suppose that $|S|<K$.  Then since $|E|=m$ and all courses have capacity 1, $M\leq |S|+|E|\leq m+|S|<m+K$, a contradiction.  Hence $|S|\geq K$.  We now claim that $S$ is an independent set in $G$.  For, suppose that $v_i$ and $v_j$ are two adjacent vertices in $G$ that are both in $S$.  Clearly $(a_i,v_i)\in M$ and $(a_j,v_j)\in M$.  But since $v_i$ and $v_j$ are adjacent in $G$, it is then impossible for both $a_i$ and $a_j$ to meet their prerequisites on $v_i$ and $v_j$ in $I$, respectively, a contradiction.
\end{proof}
\begin{cor}\label{t_MAX-POM-CAPR-hard}
{\sc max pom capr-d} is NP-hard, even if each applicant has capacity 4 and each course has capacity 1, and the prerequisites are the same for all applicants.
\end{cor}
\begin{proof}
In the proof of Theorem \ref{t_MAX-CAPR-hard}, the matching $M$ in $I$ constructed from an independent set $S$ in $G$ is in fact Pareto optimal.  To see this, let $\sigma$ be an ordering of the applicants such that every applicant in $A'$ precedes every applicant in $A\backslash A'$.  Let $M$ be the result of running Algorithm {\sf SM-CAPR} relative to the ordering $\sigma$.  It follows by Theorem \ref{thm_alg} that $M$ is a POM in $I$.  The remainder of the proof of Theorem \ref{t_MAX-CAPR-hard} can then be used to show that {\sc max pom capr-d} is NP-hard.
\end{proof}

We now give an alternative proof of Corollary \ref{t_MAX-POM-CAPR-hard} for the case that the constructed matching is required to be Pareto optimal.  This gives NP-hardness for stronger restrictions on prerequisites than those given by Corollary \ref{t_MAX-POM-CAPR-hard}.  Our reduction involves a transformation from the following NP-complete problem \cite{BKS03}:\\

\begin{minipage}{5.5in}
\emph{Name:} \sat\\
\emph{Instance:} a Boolean formula $B$, where each clause in $B$ has size three, and each variable occurs exactly twice as an unnegated literal and exactly twice as a negated literal in $B$.\\
\emph{Question:} is $B$ satisfiable?
\end{minipage}
\smallskip

\begin{thm}\label{t_hard3}
{\sc max pom capr-d} is NP-hard, even if each course has at most one prerequisite, and the prerequisites are the same for all applicants.
\end{thm}
\begin{proof}
We firstly remark that, in view of Theorem \ref{t_POM}, it is not known whether {\sc max pom capr-d} belongs to NP.  We show NP-hardness for this problem via a reduction from \sat\ as defined above.

Let $B$ be an instance of \sat, where $V=\{v_1,v_2,\dots,v_n\}$ is the set of variables and $C=\{c_1,c_2,\dots,c_m\}$ is the set of clauses.  We construct an instance $I$ of {\sc max pom capr-d} as follows.  Let $X\cup Y$ be the set of courses, where $X=\{x_i^1,x_i^2,\bar{x}_i^1,\bar{x}_i^2 : 1\leq i\leq n\}$ and $Y=\{y_i^1,y_i^2 : 1\leq i\leq n\}$.  The courses in $X$ correspond to the first and second occurrences of $v_i$ and $\bar{v}_i$ in $B$ for each $i$ ($1\leq i\leq n$).  Let $A\cup G$ be the set of applicants, where $A=\{a_j : 1\leq j\leq m\}$ and $G=\{g_i^1,g_i^2 : 1\leq i\leq n\}$.  Each course has capacity 1.  Each applicant in $A$ has capacity 1, whilst each applicant in $G$ has capacity 2.  For each $i$ ($1\leq i\leq n$), define the prerequisite $y_i^1\to y_i^2$, which is the same for all applicants.  For each $j$ ($1\leq j\leq m$) and for each $s$ ($1\leq s\leq 3$), let $x(c_j^s)$ denote the $X$-course corresponding to the literal appearing at position $s$ of clause $c_j$ in $B$.  For example if the second position of clause $c_5$ contains the second occurrence of literal $\bar{v}_i$, then $x(c_5^2)=\bar{x}_i^2$.  The preference lists of the applicants are as follows:
\[
\begin{array}{rll}
P(a_j): & x(c_j^1),x(c_j^2) , x(c_j^3) & (1\leq j\leq m)\\
P(g_i^1): & y_i^1 , y_i^2 , x_i^1, x_i^2 & (1\leq i\leq n)\\
P(g_i^2): & y_i^1, y_i^2 , \bar{x}_i^1, \bar{x}_i^2 & (1\leq i\leq n)
\end{array}
\]
We claim that $B$ has a satisfying truth assignment if and only if $I$ has a POM of size $m+4n$.

For, suppose that $f$ is a satisfying truth assignment for $B$.  We form a matching $M$ in $I$ as follows.  For each $i$ ($1\leq i\leq n)$, if $f(v_i)$={\tt true} then add the pairs $(g_i^1,y_i^1)$, $(g_i^1,y_i^2)$, $(g_i^2,\bar{x}_i^1)$, $(g_i^2,\bar{x}_i^2)$ to $M$.  On the other hand if $f(v_i)$={\tt false} then add the pairs $(g_i^1,x_i^1)$, $(g_i^1,x_i^2)$, $(g_i^2,y_i^1)$, $(g_i^2,y_i^2)$ to $M$.  For each $j$ ($1\leq j\leq m$), at least one literal in $c_j$ is true under $f$.  Let $s$ be the minimum integer such that the literal at position $s$ of $c_j$ is true under $f$.  Course $x(c_j^s)$ is still unmatched by construction; add $(a_j,x(c_j^s))$ to $M$.  It may be verified that $M$ is a POM of size $m+4n$ in $I$.

Conversely suppose that $M$ is a POM in $I$ of size $m+4n$.  Then the cardinality of $M$ implies that every applicant is full in $M$.  We firstly show that, for each $i$ ($1\leq i\leq n)$, either $\{(g_i^1,y_i^1),(g_i^1,y_i^2)\}\subseteq M$ or $\{(g_i^2,y_i^1),(g_i^2,y_i^2)\}\subseteq M$.  Suppose this is not the case.  As a consequence of the prerequisites, if $(g_i^r,y_i^1)\in M$ for some $i$ ($1\leq i\leq n$) and $r\in \{1,2\}$, then $(g_i^r,y_i^2)\in M$.  Suppose now that $(g_i^r,y_i^2)\in M$ for some $r\in \{1,2\}$, but $(g_i^r,y_i^1)\notin M$.  Let $M'$ be the matching obtained from $M$ by removing any assignee of $g_i^r$ worse than $y_i^2$ (if such an assignee exists) and adding $(g_i^r,y_i^1)$ to $M$.  Then $M'$ dominates $M$, a contradiction.  Now suppose that $y_i^2$ is unmatched in $M$.  Let $M'$ be the matching obtained from $M$ by removing any assignee of $g_i^1$ worse than $y_i^2$ (if such an assignee exists) and adding $(g_i^1,y_i^r)$ to $M$ ($r\in \{1,2\}$).  Then $M'$ dominates $M$, a contradiction.  Thus the claim is established.  It follows that for each $i$ ($1\leq i\leq n$), either $g_i^1$ is matched in $M$ to two members of $X$ and $g_i^2$ is matched in $M$ to two members of $Y$, or vice versa.

Now create a truth assignment $f$ in $B$ as follows.  For each $i$ ($1\leq i\leq n$), if $(g_i^1,y_i^1)\in M$, set $f(v_i)$={\tt true}, otherwise set $f(v_i)$={\tt false}.  We claim that $f$ is a satisfying truth assignment for $B$.  For, let $j$ ($1\leq j\leq m$) be given.  Then $(a_j,x(c_j^s))\in M$ for some $s$ ($1\leq s\leq 3$).  If $x(c_j^s)=x_i^r$ for some $i$ ($1\leq i\leq n$) and $r$ ($r\in \{1,2\}$) then $f(v_i)=T$ by construction.  Similarly if $x(c_j^s)=\bar{x}_i^r$ for some ($1\leq i\leq n$) and $r$ ($r\in \{1,2\}$) then $f(v_i)=F$ by construction.  Hence $f$ satisfies $B$. 
\end{proof}

The next example shows that the difference in cardinalities between POMs may be arbitrary, and that {\sf SM-CAPR} is not in general a constant-factor approximation algorithm for {\sc max pom capr}. 

\begin{example} Consider a {\sc capr} instance $I$ in which $A=\{a_1,a_2\}$ and $C=\{c_1,c_2,\dots,c_n\}$ for some $n\geq 1$. Let the preferences of the applicants be
$$P(a_1): c_1, c_2,\dots, c_n \qquad  P(a_2): c_n$$
and let $c_i\rightarrow c_{i+1}$ for each applicant ($1\leq i\leq n-1$).  Assume that $a_1$ has capacity $n$, whilst the capacity of $a_2$  and the capacity of every course is 1.
 
There are two POMs in $I$: if {\sf SM-CAPR} is executed relative to a policy in which $a_1$ is first then we obtain the POM $M_1$ that assigns all the $n$ courses to $a_1$ and nothing to $a_2$; if instead $a_2$ is first, we obtain the POM $M_2$ that assigns nothing to $a_1$ and the single course $c_n$ to $a_2$.  Hence executing {\sf SM-CAPR} relative to different policies can give rise to POMs with arbitrarily large difference in cardinality.  It follows that {\sf SM-CAPR} is not in general a constant-factor approximation algorithm for {\sc max pom capr}. However, notice that in this example the cardinality of the down-set of each course is not bounded by a constant; enforcing such a condition could improve the approximation possibilities.   \qed
\end{example}

\section{Alternative prerequisites}\label{s_alt}
In this section we turn our attention to {\sc caapr}, the analogue of {\sc capr} in which prerequisites need not be compulsory but may be presented as alternatives.  We will show that, in contrast to the case for {\sc capr}, finding a POM is NP-hard, under either additive or lexicographic preferences.

Recall that as {\sc capr} is a special case of {\sc caapr}, Lemma \ref{t_add} implies that finding a most-preferred bundle of courses under additive preferences is NP-hard.  Now we prove a similar result for lexicographic preferences.  

\begin{lem}\label{t_altern}
The problem of finding the most-preferred feasible bundle of courses of a given applicant with lexicographic preferences in {\sc caapr} is NP-hard.
\end{lem}
\begin{proof} 
We reduce from {\sc vc-d},  the decision version of {\sc vc},  which is the problem of finding a vertex cover of minimum size in a given graph. {\sc vc-d} is NP-complete \cite{GJ79}.
Let $\langle G,K\rangle$ be an instance of {\sc vc-d}, where $G=(V,E)$ is a graph and $K$ is a positive integer.  Assume that $V=\{v_1,v_2,\dots,v_n\}$ and $E=\{e_1,e_2,\dots,e_m\}$.
We construct an instance $I$ of {\sc caapr} as follows.  Let the set of courses be $V\cup E\cup \{b\}$ (again, we use the same notation for vertices and edges in $G$ as we do for courses in $I$, but no ambiguity should arise.)  There is a single applicant $a_1$ with capacity $m+K+1$ whose preference list is as follows:
$$P(a_1): b, e_1,e_2,\dots,e_m,v_1,v_2,\dots,v_n.$$
Course $b$ has a single compulsory prerequisite course $e_1$, and each course $e_j$ $(2\leq j\leq m-1)$ has a single compulsory prerequisite course $e_{j+1}$. Moreover,
all the  $E$-courses have (alternative) prerequisites; namely, for any $j$ ($1\leq j\leq m$), if course $e_j$ corresponds to the edge $e_j=\{v_i,v_k\}$ then $e_j\mapsto_{a_1} \{v_i,v_k\}$. We claim that $G$ admits a vertex cover of size at most $K$ if and only if $I$ admits a matching in which $a_1$ is assigned course $b$.

For, suppose that $G$ admits a vertex cover $S$ where $|S|\leq K$.  Form a matching $M$ by assigning $a_1$ to the bundle $B=\{b\}\cup E\cup S$.  Then $B$ is a feasible bundle of courses for $a_1$, and $b\in B$.

Conversely, suppose $I$ admits a matching $M$ in which $a_1$ is assigned a bundle $B$ containing course $b$. Then, due to the prerequisites, $B$ must contain all $E$-courses and for each course in $e_j\in E$, $B$ must contain some course in $v_i\in V$ that corresponds to a vertex incident to $e_j$.  Let $S=B\cap V$.  Clearly $S$ is a vertex cover in $G$, and as $q(a_1)=m+K+1$, it follows that $|S|\leq K$.

A polynomial-time algorithm for finding the most-preferred bundle of courses for $a_1$ can then be used to decide whether $I$ admits a matching in which $a_1$ is assigned $b$, hence the result.
\end{proof}

As noted in Section \ref{sec:comp-add}, in the case of just one applicant $a_1$, a matching $M$ is a POM if and only if $a_1$ is assigned in $M$ her most-preferred feasible bundle of courses. Hence Lemma \ref{t_altern} implies the following assertion. 

\begin{thm}
Given an instance of {\sc caapr} the problem of finding a POM is NP-hard.  The result holds under either additive or lexicographic preferences.
\end{thm}

We finally remark that, since {\sc capr} is a special case of {\sc caapr}, Theorem \ref{t_POM} implies that the problem of determining whether a given matching $M$ in an instance of {\sc caapr} is a POM is co-NP-complete for lexicographic preferences (and also for additive preferences by the remark following Theorem \ref{t_POM}).

\section{Corequisites}\label{s_coreq}
In this section we focus on {\sc cacr}, the extension of {\sc ca} involving corequisite courses.  As in the case of {\sc capr}, we will show that finding a POM in the presence of additive preferences is NP-hard.  Thus the majority of our attention is concerned with lexicographic preferences.  In this case we show how to modify the sequential mechanism 
in order to obtain a polynomial-time algorithm for finding a POM in the {\sc cacr} case.  Moreover we show that in {\sc cacr}, the problem of finding a maximum cardinality POM is very difficult to approximate.

We begin with additive preferences.  A simple modification of the proof of Lemma \ref{t_add} (ensuring that, for each $i$ ($1\leq i\leq n$), $c_i\leftrightarrow d_i^r$ for each $r$ ($1\leq r\leq w_i-1$)) gives the following analogue of Theorem \ref{thm:findPOMcapr}.
\begin{thm}\label{t_add_co} 
Given an instance of {\sc cacr} with additive preferences, the problem of finding a POM is NP-hard.
\end{thm}

In view of Theorem \ref{t_add_co}, in the remainder of this section we assume that preferences are lexicographic.  In this case we can find a POM efficiently by dealing with the corequisites as follows.  Let $I$ be an instance of {\sc cacr}, let $\sigma$ be a policy in $I$, and assume the notation defined in Section \ref{sec:coreq}.  We lose no generality by supposing that each applicant either finds all the courses in one equivalence class  $C^k$ acceptable, or none of them.  Replace all the courses in $C^k$ by a single \emph{supercourse} $c^k$ such that $q(c^k)=\min \{q(c_j) : c_j\in C^k\}$.
For any applicant $a_i\in A$ who finds all courses in $C^k$ acceptable, remove all such courses from $a_i$'s list and replace them by $c^k$; since preferences are lexicographic, the position of $c^k$ in the modified preference list of $a_i$ is the position of the most-preferred course of $C^k$ in her original list.  Let $I'$ denote the {\sc cacr} instance obtained from $I$ by using this transformation.

The sequential mechanism for {\sc cacr} can be executed on $I'$ as follows. The mechanism works according to a given policy $\sigma$ in stages. In one stage, the applicant $a_i$ who next has her turn according to $\sigma$ chooses her most-preferred supercourse $c^k$ to which she has not yet applied. Applicant $a_i$ is assigned to $c^k$ if two conditions are fulfilled: (i) the number of courses assigned to $a_i$ so far plus the cardinality of $C^k$ does not exceed $q(a_i)$, and (ii) each course $c_j\in C^k$ still has a free slot. If this is not possible, at the same stage $a_i$ applies to her next supercourse until she is either assigned some supercourse or her preference list is exhausted.  Once the whole process terminates, let $M'$ be the assignment of applicants to supercourses in $I'$ and construct the following assignment $M$ in $I$ from $M'$:
\begin{equation}
M=\{(a_i,c_j) : a_i\in A\wedge c_j\in C^k\wedge (a_i,c^k)\in M'\}. \label{eq}
\end{equation}

Let {\sf SM-CACR} denote the mechanism that constructs $M$ from $I$ and $\sigma$.  {\sf SM-CACR} always yields a POM and runs in polynomial time; the proof is very similar to that of Theorem \ref{thm_alg}.  In fact we can go a step further and generalise the mechanism to the case of {\sc cacpr}, the extension of {\sc ca} in which there are both prerequisite and corequisite courses (thus all prerequisites are defined in terms of supercourses).  Let us denote by {\sf SM-CACPR} the mechanism {\sf SM-CAPR} with the following  modifications:
\begin{itemize}
\item[(i)] $M$ should be replaced by $M'$ everywhere except in line \ref{line:return}, where $M'$ is a matching of applicants to supercourses;
\item[(ii)] every occurrence of ``course'' should be replaced by ``supercourse''; likewise $c_j$ (resp.\ $c_k$) should be replaced by $c^j$ (resp.\ $c^k$);
\item[(iii)] in line \ref{line:under}, ``$c_j$ is undersubscribed'' should be replaced by ``if each course in $c^j$ is undersubscribed'';
\item[(iv)] in line \ref{line:full}, ``$c_k$'' is full should be replaced by ``some course in $c^k$ is full'';
\item[(v)] in line \ref{line:quota}, $|M(a_i)|$ is interpreted as $\sum_{c^j\in M(a_i)} |C^j|$ and $|S|$ is interpreted as $\sum_{c^j\in S} |C^j|$;
\item[(vi)] after line \ref{line:penultimate}, $M$ should be obtained from $M'$ as per Equation \ref{eq}.
\end{itemize}
We can then arrive at the following result, whose proof is a straightforward extension of that of Theorem \ref{thm_alg} and is omitted.
\begin{thm}
\label{thm:cacpr}
Algorithm {\sf SM-CACPR} produces a POM for a given instance $I$ of {\sc cacpr} and for a given policy $\sigma$ in $I$.  The complexity of the algorithm is $O(N+n_2(L+\Delta))$, where $N$ is the sum of the applicants' capacities, $n_2$ is the number of courses, $L$ is the total length of the applicants' preference lists and $\Delta$ is the total number of immediate prerequisites of each course $c_j$ in $\to_{a_i}$, taken over each applicant $a_i$.
\end{thm}

In the {\sc cacr} model as defined in Section \ref{sec:coreq}, corequisite constraints are common to all applicants.  In this setting, and after the modification described prior to Theorem \ref{thm:cacpr}, in which courses are merged into supercourses, {\sc cacr} becomes equivalent to {\sc cap}, the extension of {\sc ca} with price-budget constraints described in \cite{CEFMMP14}.  For an instance $I$ of {\sc cap}, it is known that for each POM $M$ in $I$, there exists a policy $\sigma$ such that executing {\sf SM-CACR} relative to $\sigma$ produces $M$ \cite[Theorem 3]{CEFMMP14}.
Our next example presents an observation about the behaviour of {\sf SM-CACR} if we extend it to the variant of {\sc cacr} in which corequisite constraints are specific to individual applicants.
\begin{example}
\label{ex:cacr}
The {\sf SM-CACR} mechanism can be extended without difficulty to the variant of {\sc cacr} (considered in this example only) in which corequisites can be applicant-specific.  However it is no longer true that the mechanism is capable of reaching all POMs relative to a suitable policy, as we now illustrate.  Consider a {\sc cacr} instance with two applicants and three courses.  Suppose that each applicant has capacity 2, and that each course has capacity 1.  Assume that the applicants have the following preference lists:
\[
\begin{array}{ll}
P(a_1): & c_1, c_2, c_3\\
P(a_2): & c_2, c_1, c_3
\end{array}
\]
Assume that each applicant has as corequisites the first and last courses on her list.  Then {\sf SM-CACR} will return the matching $\{(a_i,c_i),(a_i,c_3)\}$ if the first applicant in the policy is $a_i$ ($i\in \{1,2\}$).  However the matching $M=\{(a_1,c_2),(a_2,c_1)\}$ is also Pareto optimal and cannot be obtained by {\sf SM-CACR}. \qed
\end{example}



Given an instance of {\sc cap}, the problem of finding a maximum cardinality POM is NP-hard \cite[Theorem 7]{CEFMMP14}.  
Using the connection between {\sc cacr} and {\sc cap} described prior to Example \ref{ex:cacr}, the same is therefore true for {\sc max pom cacr}, the problem of finding a maximum cardinality POM, given an instance of {\sc cacr}.  We now strengthen this result by showing that {\sc max pom cacr} is very difficult to approximate.
\begin{thm} \label{t_cacr}
{\sc max pom cacr} is NP-hard and not approximable within a factor of $N^{1-\varepsilon}$, for any $\varepsilon>0$, unless P=NP, where $N$ is the total capacity of the applicants.
\end{thm}
\begin{proof}
Let $\varepsilon>0$ be given.  Let $B$ be an instance of \sat, where $V=\{v_1,v_2,\dots,v_n\}$ is the set of variables and $C=\{c_1,c_2,\dots,c_m\}$ is the set of clauses.  Let $\beta=\left\lceil\frac{2}{\varepsilon}\right\rceil$ and let $\alpha=n^\beta$.

We form an instance $I$ of {\sc cacr} as follows.  Let $X\cup Y\cup Z$ be the set of courses, where $X=\{x_i^1,x_i^2,\bar{x}_i^1,\bar{x}_i^2 : 1\leq i\leq n\}$, $Y=\{y_i^1,y_i^2 : 1\leq i\leq n\}$ and $Z=\{z_1,z_2,\dots,z_D\}$, where $D=6n(\alpha-1)+1$.  The courses in $X$ correspond to the first and second occurrences of $v_i$ and $\bar{v}_i$ in $B$ for each $i$ ($1\leq i\leq n$).  Let $A\cup G\cup \{b,h\}$ be the set of applicants, where $A=\{a_j : 1\leq j\leq m\}$ and $G=\{g_i^1,g_i^2 : 1\leq i\leq n\}$.  Each course has capacity 1.  Each applicant in $A$ has capacity 1, each applicant in $G$ has capacity 2, $h$ has capacity $2n-m$ and $b$ has capacity $D$.

For each $i$ ($1\leq i\leq n$), courses $y_i^1$ and $y_i^2$ are corequisites.  Also all the courses in $Z$ are corequisites.  For each $j$ ($1\leq j\leq m$) and for each $s$ ($1\leq s\leq 3$), $x(c_j^s)$ is as defined in the proof of Theorem \ref{t_hard3}.  The preference lists of the applicants are as follows:
\[
\begin{array}{rll}
P(a_j): & x(c_j^1), x(c_j^2), x(c_j^3) & (1\leq j\leq m)\\
P(g_i^1): & y_i^1, y_i^2, x_i^1, x_i^2 & (1\leq i\leq n)\\
P(g_i^2): & y_i^1, y_i^2, \bar{x}_i^1, \bar{x}_i^2 & (1\leq i\leq n)\\
P(h) : & [X] \\
P(b) : & [X], [Z]
\end{array}
\]
In the preference lists of $h$ and $b$, the symbols $[X]$ and $[Z]$ denote all members of $X$ and $Z$ listed in arbitrary strict order, respectively. 
In $I$ the total capacity of the applicants, denoted by $N$, satisfies $N=D+6n$.  We claim that if $B$ has a satisfying truth assignment then $I$ has a POM of size $D+6n$, whilst if $B$ does not have a satisfying truth assignment then any POM in $I$ has size at most $6n$.

For, suppose that $f$ is a satisfying truth assignment for $B$.  We form a matching $M$ in $I$ as follows.  For each $i$ ($1\leq i\leq n)$, if $f(v_i)$={\tt true} then add the pairs $(g_i^1,y_i^1)$, $(g_i^1,y_i^2)$, $(g_i^2,\bar{x}_i^1)$, $(g_i^2,\bar{x}_i^2)$ to $M$.  On the other hand if $f(v_i)$={\tt false} then add the pairs $(g_i^1,x_i^1)$, $(g_i^1,x_i^2)$, $(g_i^2,y_i^1)$, $(g_i^2,y_i^2)$ to $M$.  For each $j$ ($1\leq j\leq m$), at least one literal in $c_j$ is true under $f$.  Let $s$ be the minimum integer such that the literal at position $s$ of $c_j$ is true under $f$.  Course $x(c_j^s)$ is still unmatched by construction; add $(a_j,x(c_j^s))$ to $M$.  There remain $2n-m$ courses in $X$ that are as yet unmatched in $M$; assign all these courses to $h$.  Finally assign all courses in $Z$ to $b$ in $M$.  It may be verified that $M$ is a POM of size $D+6n$ in $I$.

Now suppose that $f$ admits no satisfying truth assignment.  Let $M$ be any POM in $I$.  We will show that $|M|\leq 6n$.  Suppose not.  Then $|M|>6n$ and the only way this is possible is if at least one course in $Z$ is matched in $M$.  But only $b$ can be assigned members of $Z$ in $M$, and since all pairs of courses in $Z$ are corequisites, it follows that $M(b)=Z$.

We next show that, for each $i$ ($1\leq i\leq n)$, either $\{(g_i^1,y_i^1),(g_i^1,y_i^2)\}$ $\subseteq M$ or $\{(g_i^2,y_i^1),$ $(g_i^2,y_i^2)\}$ $\subseteq M$.  Suppose this is not the case for some $i$ ($1\leq i\leq n$).  As a consequence of the corequisite restrictions on courses in $Y$, $y_i^1$ and $y_i^2$ are unmatched in $M$.  Let $M'$ be the matching obtained from $M$ by deleting any assignee of $g_i^1$ worse than $y_i^2$ (if such an assignee exists) and by adding $(g_i^1,y_i^1)$ and $(g_i^1,y_i^2)$ to $M$.  Then $M'$ dominates $M$, a contradiction.

We claim that each course in $X$ is matched in $M$.  For, suppose that some course $x\in X$ is unmatched.  Then let $M'$ be the matching obtained from $M$ by unassigning $b$ from all courses in $Z$, and by assigning $b$ to $x$.  Then $M'$ dominates $M$, a contradiction.

It follows that every course in $X\cup Y$ is matched in $M$.  Since $|X\cup Y|=6n$ and the applicants in $A\cup G\cup \{h\}$ have total capacity $6n$, every applicant in $A\cup G\cup \{h\}$ is full.

Create a truth assignment $f$ in $B$ as follows.  For each $i$ ($1\leq i\leq n$), if $(g_i^1,y_i^1)\in M$, set $f(v_i)$={\tt true}, otherwise set $f(v_i)$={\tt false}.  We claim that $f$ is a satisfying truth assignment for $B$.  For, let $j$ ($1\leq j\leq m$) be given.  Then $(a_j,x(c_j^s))\in M$ for some $s$ ($1\leq s\leq 3$).  If $x(c_j^s)=x_i^r$ for some $i$ ($1\leq i\leq n$) and $r$ ($r\in \{1,2\}$) then $f(v_i)$={\tt true} by construction.  Similarly if $x(c_j^s)=\bar{x}_i^r$ for some ($1\leq i\leq n$) and $r$ ($r\in \{1,2\}$) then $f(v_i)$={\tt false} by construction.  Hence $f$ satisfies $B$, a contradiction.

Thus if $B$ is satisfiable then $I$ admits a POM of size $D+6n=6n(\alpha-1)+1+6n>6n\alpha$.  If $B$ is not satisfiable then any POM in $I$ has size at most $6n$.  Hence an $\alpha$-approximation algorithm for {\sc max pom cacr} implies a polynomial-time algorithm to determine whether $B$ is satisfiable, a contradiction to the NP-completeness of \sat.

It remains to show that $N^{1-\varepsilon}\leq \alpha$.  On the one hand, $N=6n+D=6n\alpha+1\leq 7n\alpha=7n^{\beta+1}$.  Hence $n^\beta\geq N^{\frac{\beta}{\beta+1}}7^{-\frac{\beta}{\beta+1}}$.  On the other hand, $N=6n\alpha+1\geq \alpha=n^\beta\geq 7^\beta$ as we may assume, without loss of generality, that $n\geq 7$.  It follows that $7^{-\frac{\beta}{\beta+1}}\geq N^{-\frac{1}{\beta+1}}$.  Thus
\[\alpha=n^\beta\geq N^{\frac{\beta}{\beta+1}}7^{-\frac{\beta}{\beta+1}}\geq N^{\frac{\beta}{\beta+1}}N^{-\frac{1}{\beta+1}}=N^{\frac{\beta-1}{\beta+1}}=N^{1-\frac{2}{\beta+1}}\geq N^{1-\varepsilon}.\vspace{-7mm}\]
\end{proof}

\section{Open problems and directions for future research }\label{s_open}

We would like to conclude with several open problems and directions for future research.
\begin{enumerate}
\item {\bf Refining the boundary between efficiently solvable and hard problems.} In the proofs of the NP-hardness and inapproximability results in this paper 
we had some applicants whose preference lists were not complete and/or whose capacity was not bounded by a constant. Will the hardness results still be valid if there are no unacceptable courses and all capacities are bounded? These problems also call for a more detailed multivariate complexity analysis. It might be interesting to determine whether restricting some other parameters, e.g., the lengths of preference lists, may make the problems tractable.

Other problems for which the computational complexity has not yet been resolved include the complexity of determining whether a matching is Pareto optimal, given (i) an instance of {\sc cacr}, or (ii) an instance of {\sc capr} where the pre-requisites are the same for all applicants (this is not the case in the {\sc capr} instance constructed by the proof of Theorem \ref{t_POM}).
\item
{\bf Indifferences in the preference lists.} In this paper, we assumed that all the preferences are strict.
If preference lists contain ties, sequential mechanisms have to be carefully modified to ensure Pareto optimality. Polynomial-time algorithms for finding a Pareto optimal matching in the presence of ties have been given in the context of {\sc hat} (the extension of {\sc ha} where preference lists may include ties) by Krysta et al.\ \cite{KMRZ14} and in its many-to-many generalisation {\sc cat} (the extension of {\sc ca} where preference lists may include ties) by Cechl\'arov\'a et al.\ \cite{CEFMMM16}.  However as far as we are aware, it remains open to extend these algorithms to the cases of {\sc capr} and {\sc cacr} where preference lists may include ties.
\item
{\bf Strategic issues.} 
By a standard argument, one can ensure that the sequential mechanism that lets each applicant on her turn choose her entire most-preferred bundle of courses (i.e., the serial dictatorship mechanism) is strategy-proof even in the case of prerequisites.  However, serial dictatorship may be very unfair, as the first dictator may grab all the courses and leave nothing for the rest of the applicants. Let us draw  the reader's attention to several economic papers that highlight the special position of serial dictatorship among the mechanisms for allocation of multiple indivisible goods: serial dictatorship is the only allocation rule that is Pareto efficient,  strategy-proof and fulfils some additional properties, namely non-bossiness and citizen sovereignty \cite{Pap01}, and population monotonicity or consistency \cite{KM01}. We were not able to obtain a similar characterization of serial dictatorship for {\sc capr}. 
 
As far as the general sequential mechanism is concerned, a recent result by Hosseini and Larson \cite{HL15} shows that no sequential mechanism that allows interleaving policies (i.e., in which an $a_i$ is allowed to pick courses more than once, and between two turns of $a_i$ another applicant has the right to pick a course) is strategy-proof, even in the simpler case without any prerequisites.  It immediately follows that {\sf SM-CAPR} is not strategy-proof.  However, it is not known whether a successful manipulation can be computed efficiently.  Further, we have shown that in {\sc capr}, not all POMs can be obtained by a sequential mechanism.  We leave it as an open question whether a strategy-proof and Pareto optimal mechanism other than serial dictatorship exists in {\sc capr}.

\end{enumerate}

\end{document}